\documentclass[9pt,letterpaper]{article}

\usepackage{amsmath, amsthm, amssymb}
\usepackage[hmargin=3.4cm,vmargin=3.2cm]{geometry}
\usepackage{hyperref}
\usepackage{authblk}
\usepackage[runin]{abstract}
\usepackage{sectsty}
\usepackage[font=small,labelfont=small]{caption}

\newtheorem{definition}{Definition}
\newtheorem{lemma}{Lemma}
\newtheorem{theorem}{Theorem}
\newtheorem{proposition}[theorem]{Proposition}

\theoremstyle{remark}

\title{\Large\bf Tight Bounds for $L_p$ Samplers, Finding Duplicates in
  Streams, and Related Problems}

\author[1]{Hossein Jowhari}
\author[1]{Mert Sa\u{g}lam}
\affil[1]{Simon Fraser University, Burnaby, Canada}
\author[1,2]{G\'abor Tardos}
\affil[2]{R\'enyi Institute of Mathematics, Budapest, Hungary}
\date{}

\DeclareMathOperator*{\E}{\mathbb{E}}

\DeclareMathOperator{\err}{Err}

\DeclareMathOperator{\URn}{UR^n}
\DeclareMathOperator{\polylog}{polylog}
\DeclareMathOperator{\poly}{poly}

\begin{document}
\maketitle
\sectionfont{\fontsize{11}{11}\selectfont}
\subsectionfont{\fontsize{10}{11}\selectfont}
\paragraphfont{\fontsize{9}{11}\selectfont}
\subparagraphfont{\fontsize{9}{11}\selectfont}
\minisecfont{\fontsize{9}{11}\selectfont}
\fontsize{9}{11}\selectfont
%
%
\setlength{\abstitleskip}{-\absparindent}
\thispagestyle{empty}
\vspace{-7.5mm}
\begin{abstract}
In this paper, we present near-optimal space bounds for $L_p$-samplers.
Given a stream of updates (additions and subtraction) to the coordinates of
an underlying vector $x \in \mathbb R^n$, a perfect $L_p$ sampler outputs
 the $i$-th coordinate with probability $|x_i|^p/\|x\|_p^p$.
In SODA 2010, Monemizadeh and Woodruff showed polylog space upper bounds for
 approximate $L_p$-samplers and demonstrated various applications of them.
Very recently, Andoni, Krauthgamer and Onak improved the
upper bounds and gave a $O(\epsilon^{-p}\log^3 n)$ space $\epsilon$
relative error and constant failure rate $L_p$-sampler for $p \in [1,2]$.
In this work, we give another such algorithm requiring only
$O(\epsilon^{-p}\log^2n)$ space for $p \in (1,2)$. For $p \in (0,1)$,
our space bound is  $O(\epsilon^{-1}\log^2n)$, while for the $p=1$ case we have
an $O(\log(1/\epsilon)\epsilon^{-1}\log^2n)$ space algorithm.
We also give a $O(\log^2 n)$ bits zero relative error
$L_0$-sampler, improving the $O(\log^3 n)$ bits algorithm due to  Frahling, Indyk and Sohler.


As an application of our samplers, we give better upper bounds
for the problem of finding duplicates in data streams. In case
the length of the stream is longer than the
alphabet size, $L_1$ sampling gives us an $O(\log^2 n)$ space
algorithm, thus improving the previous $O(\log^3 n)$ bound due to Gopalan
and Radhakrishnan.

In the second part of our work, we prove an $\Omega(\log^2 n)$ lower bound for
sampling
from 0, $\pm1$ vectors (in this special case, the parameter $p$ is not relevant for
$L_p$ sampling).
 This matches the space of our sampling algorithms for constant
$\epsilon>0$. We also prove tight space lower bounds for
the finding duplicates and
heavy hitters problems. We obtain these lower bounds
using reductions from the communication complexity problem augmented indexing.
\end{abstract}
  
\newpage
\setcounter{page}{1}
%
%
\section{Introduction}\label{sec:intro}


Sampling has become an indispensable tool in analysing massive data sets,
and particularly in processing data streams. In the past decade, 
several sampling techniques have been proposed and studied for the data stream
model \cite{BabcockDM02,DuffieldLT07,BravermanOZ09,CormodeMYZ10,MonemizadehW10,AndoniKO10}.
In this work, we study {\em $L_p$-samplers}, a new variant of
space efficient samplers for data streams that
was introduced by Monemizadeh and
Woodruff in \cite{MonemizadehW10}. 
 Roughly speaking, given a stream of updates (additions and subtraction) 
to the coordinates of an underlying vector $x \in \mathbb R^n$,
an $L_p$-sampler processes the stream and 
outputs a sample coordinate of $x$ where the
$i$-th coordinate is picked with probability proportional to
$|x_i|^p$.

In \cite{MonemizadehW10}, it was observed that $L_p$-samplers lead to 
alternative algorithms for many known streaming problems, including  
heavy hitters and frequency moment estimation. Here in this paper, we
focus on a specific application, namely finding duplicates in long streams;
although our $L_p$ samplers  work and often give better space performance
for all applications listed in \cite{MonemizadehW10}.
We refer the reader to \cite{MonemizadehW10} and \cite{AndoniKO10} 
for further applications of $L_p$-samplers.



Observe that we allow both negative and positive updates to the coordinates of the underlying vector.
In the case where only positive updates are allowed and $p=1$, the problem is well understood.
    For instance the classical reservoir sampling \cite{Knuth69} from the 60's
    (attributed to Alan G.~Waterman) gives a simple 
    solution as follows. Given a pair $(i,u)$, 
    indicating an addition of $u$ to the  $i$-th coordinate of 
    the underlying vector $x$, the sampler having maintained $s$, the sum of the 
    updates seen so far, replaces
     the current sample with $i$ with probability $u/s$, otherwise does nothing 
     and moves to the next update. It is easy to verify that this is a perfect
     $L_1$-sampler and the space usage is only $O(1)$ words.
    
   With the presence of negative updates, sampling becomes a non-trivial problem. In this case,
   it is not clear at all how to maintain samples without keeping track of the updates
    to the individual coordinates. In fact, the question
    regarding the mere existence of such samplers was raised few years ago by Cormode,
    Muthukrishnan, and Rozenbaum in \cite{CormodeMR05}. Last year in SODA 2010, 
    Monemizadeh and Woodruff \cite{MonemizadehW10} answered this question
    affirmatively, however in an approximate sense. Before stating 
     their results we give a formal definition of 
    $L_p$-samplers.

\begin{definition}
Let $x \in \mathbb{R}^n$ be a non-zero vector. For
$p>0$ we call the {\em $L_p$ distribution} corresponding to $x$ the
distribution on $[n]$ that takes $i$ with probability
$$\frac{|x_i|^p}{\|x\|_p^p},$$
with $\|x\|_p=(\sum_{i=1}^n|x_i|^p)^{1/p}$. For $p=0$,
the $L_0$ distribution
corresponding to $x$ is the uniform distribution over the non-zero coordinates
of $x$.
\end{definition}

We call a streaming algorithm a {\em perfect $L_p$-sampler} if it outputs an
index according to this distribution and fails only if $x$ is the zero
vector. An approximate $L_p$-sampler may fail but the distribution of its
output should be close to the $L_p$ distribution. In particular, we speak of an $\epsilon$ relative
error $L_p$-sampler if, conditioned on no failure, it outputs the index $i$
with probability $(1\pm\epsilon)|x_i|^p/\|x\|_p^p+O(n^{-c})$, where $c$ is an
arbitrary constant. For $p=0$ the corresponding formula is
$(1\pm\epsilon)/k+O(n^{-c})$, where $k$ is the number of non-zero coordinates
in $x$. Unless stated otherwise we assume
that the failure probability is at most $1/2$.

In this definition one can consider $c$ to be 2, but all existing
constructions of $L_p$-samplers work for an arbitrary $c$ with just a
constant factor increase in the space, so we will not specify $c$ in the
following and ignore errors of probability $n^{-c}$.
 
\paragraph{Previous work.} 
A zero relative error $L_0$-sampler which uses $O(\log^3 n)$ bits was shown in \cite{FrahlingIS05}. In \cite{MonemizadehW10}, the authors gave an $\epsilon$ relative error
$L_p$-sampler for $p \in [0,2]$ which uses $\poly(\epsilon^{-1},\log n)$ space. They also showed 
  a 2-pass $O(\polylog n)$ space zero relative error $L_p$-sampler for any $p\in [0,2]$. In addition
  to these, they demonstrated that $L_p$-samplers can be used as a black-box to obtain
   streaming algorithms for other problems such as $L_p$ estimation (for $p >2$), heavy hitters, 
   and cascaded norms \cite{JayramW09}.
   Unfortunately, due to the large exponents in their bounds, the
  $L_p$-samplers given there do not lead to efficient solutions for the aforementioned applications.
  
  Very recently, Andoni, Krauthgamer and Onak in 
\cite{AndoniKO10} improved the results of \cite{MonemizadehW10} 
 considerably. Through the adaptation of 
  a generic and simple method, named {\it precision sampling}, they managed
   to bring down the space upper bounds to $O(\frac1{\epsilon^p}\log^3 n)$ bits
   for $\epsilon$ relative error $L_p$-samplers for $p \in [1,2]$. Roughly speaking, the idea of precision
  sampling is to scale the input vector with random coefficients so that the $i$-th coordinate
  becomes the maximum with probability roughly proportional to $|x_i|^p$. Moreover the maximum (heavy)
   coordinate is found through a small-space heavy hitter algorithm. In more detail,
    the input vector $(x_1,\ldots,x_n)$ is scaled 
   by random coefficients $(t_1^{-1},\ldots,t_n^{-1})$, where each $t_i$ is picked uniformly 
   at random from $[0,1]$. Let $z=(x_1t_1^{-1},\ldots,x_nt_n^{-1})$ be the 
   scaled vector. 
 Here the 
   important observation is $\Pr[t_i^{-1} \ge t] =1/t$ 
   and hence, for instance,  by replacing $t$ with 
   $\|x\|_1/|x_i|$, we get $\Pr[|z_i| \ge \|x\|_1] = |x_i|/\|x\|_1$. (In the
   same manner, one can scale $x_i$ by $t_i^{-1/p}$
   instead of $t_i^{-1}$ and get a similar result for general $p$.) 
It turns out, we only need to we have a constant approximation to
$\|x\|_1$ and look for a coordinate in $z$ that has reached a
   limit of $\Omega(\|x\|_1)$. On the other hand it is shown that the heaviest coordinate in $z$ has
   a weight of $\Omega(\log^{-1}{n})\|z\|_1$ (with constant probability), and thus
   a small-space heavy hitter computation can be used to find the maximum. In particular, 
   the $L_p$-sampler of \cite{AndoniKO10} adapts the
   popular {\it count-sketch} scheme \cite{CharikarCF04} for this purpose.
   \paragraph{Our contributions.} In 
   this paper,  we  give $L_p$-samplers requiring only
$O(\epsilon^{-p}\log^2n)$ space for $p \in (1,2)$. For $p \in (0,1)$,
our space bound is  $O(\epsilon^{-1}\log^2n)$, while for the $p=1$ case we have
an $O(\log(1/\epsilon)\epsilon^{-1}\log^2n)$ space algorithm.
    In essence, our sampler follows the basic structure of the
     precision sampling method explained above. 
      However compared to \cite{AndoniKO10}, we
      do a sharper analysis of the error terms in the count-sketch, and through
    additional ideas, we manage to get rid of a log factor
      and preserve the previous dependence on $\epsilon$. 
      Roughly speaking, we use the fact that 
      the error term in the count-sketch is bounded by the $L_2$ 
      norm of the tail distribution of $z$ (the heavy coordinates do not contribute). On
      the other hand, taking 
     the distribution of the random coefficients into account, we
     bound this by $O(\|x\|_p)$, which enables us to save a 
     log factor. Additionally, to preserve
      the dependence on $\epsilon$, we have to use a slightly more 
     powerful source of randomness for choosing the scaling factors
     (in contrast with the pairwise-independence of \cite{AndoniKO10}), and 
     take care of some subtle issues regarding the conditioning on
       the error terms which are not handled in the previous work 
       (Lemma \ref{abort}).\footnote{       
       Further we note that our
algorithm not only produces a sample $i$ from the $L_p$ distribution,
but also approximates $x_i$. Similar approximation is also produced by
the $L_p$ sampler of \cite{AndoniKO10}, but they
claim to give an approximation of $|x_i|^p/||x||_p^p$.
However, this claim for $p<2$ cannot hold as it would contradict
with the $\Omega(\epsilon^{-2})$ space lower bound for estimating 
Hamming distance.}
  
     
   As $p$ approaches zero, precision sampling becomes very inefficient, as
  the random coefficients $t_i^{-1/p}$ tend to infinity. For the $p=0$ case,
 we present a completely different algorithm. 
   Briefly, our $L_0$-sampler tries to 
   detect a non-zero
   coordinate by picking random subsets of $[n]$.
   The non-zero coordinates are found
   by an exact sparse recovery procedure and Nisan's PRG \cite{Nisan} 
   is applied to decrease the randomness involved.
   Our $O(\log^2 n)$ space bound compares favorably to the previous algorithms, which use respectively $O(\log^3 n)$ space \cite{FrahlingIS05} and $\poly(\log n,\epsilon^{-1})$  space \cite{MonemizadehW10} (the latter one gives only $\epsilon$-relative error sampling).

    In Section~\ref{sec:lb}, we prove that sampling from 0, $\pm1$
   vectors requires $\Omega(\log^2 n)$ space, by a reduction from 
   the communication complexity problem augmented indexing. 
   In this special case $p$ is not relevant for $L_p$-sampling, hence 
   this shows that our $L_0$-sampling algorithm uses the optimal space up to constant 
   factors, and our $L_p$-sampler for $p\in(0,2)$ has the optimal space (up to constant factors) for $\epsilon>0$ a constant. 

Given a stream of length $n+1$ over the alphabet $[n]$, finding
duplicates problem asks to output some $a\in[n]$ that has appeared
at least twice in the stream. Observe that by the pigeon-hole principle, 
such $a$ always exists. Prior to our work, the best upper bound
for finding duplicates was due to  Gopalan
and Radhakrishnan \cite{GopalanJaikumar}, who gave a one-pass $O(\log^3
n)$ bits randomized algorithm with constant failure rate.
Here we settle the one-pass complexity of this problem by 
giving an $O(\log^2 n)$ space algorithm via a direct application of our $L_1$
sampler, and by giving an $\Omega(\log^2 n)$ lower bound afterwards. Combined with a sparse recovery procedure,
our solution also generalizes to a near-optimal $O(\log^2n+s\log n)$ space
algorithm for finding duplicates in streams of length $n-s$, improving on the $O(s\log^3 n)$ result of \cite{GopalanJaikumar}.



Finally, we prove lower bounds for the problem of finding heavy hitters in update streams, 
which is closely related to the $L_p$-sampling problem. This lower bound is also obtained 
by a reduction from the augment indexing and proves that any $L_p$ heavy 
hitters algorithm (defined in Section \ref{sec:hh}) must use 
$\Omega(\frac{1}{\phi^p}\log^2 n)$ space, even in the strict turnstile model. 
Our lower bound essentially matches the known upper bounds \cite{CormodeM05,CharikarCF04,KaneNPW} which work in the general update model. 

\paragraph{Related work.} In \cite{BabcockDM02,BravermanOZ09}, the authors have 
studied sampling from sliding windows, and the recent paper
of Cormode et al.\ \cite{CormodeMYZ10} generalizes the classical reservoir sampling 
to distributed streams. These works only support insertion streams.
 The basic idea of random scaling used in \cite{AndoniKO10} and in our paper has appeared earlier in
 the priority sampling technique \cite{DuffieldLT07,CohenDKLT09}, where the focus is
 to estimate the weight of certain subsets of a vector, defined by a sequence of positive updates.
 
Finding duplicates in streams was first considered 
in the context of detecting fraud in click streams \cite{MetwallyAA05}. 
Muthukrishnan in \cite{Muthukrishnan}
asked whether this problem can be solved in 
$O(\polylog n)$ space using a constant number of passes. In \cite{Tarui}, Tarui showed
that any $k$-pass deterministic algorithm must use $\Omega(n^{\frac1{2k-1}})$
space. 
 
Heavy hitter algorithms have been studied extensively. 
The work of Berinde et al.\
\cite{BerindeCIS09} gives tight lower bounds for heavy hitters under 
 insertion-only streams. We are not aware of similar works on
general update streams, although the recent works of \cite{BaIPW10, WoodruffJ11},
 where the authors show lower bounds for respectively approximate
sparse recovery, and Johnson-Lindenstrauss transforms (via augmented indexing) is closely related.

\paragraph{Notation.}
We write $[n]$ for the set $\{1,\ldots,n\}$.
An update stream is a sequence of tuples $(i,u)$, where $i\in [n]$ 
and $u\in\mathbb R$. The stream of updates implicitly define an $n$-dimensional
vector $x\in\mathbb R^n$ as follows. Initially, $x$ is the zero vector. An update of the form
$(i,u)$ adds $u$ to the coordinate $x_i$ of $x$ (leaving the other coordinates unchanged).
In the {\em strict turnstile model}
we are guaranteed that all coordinates of $x$ are non-negative at the end
of the stream
(although negative updates are still allowed), in the general model such
guarantee does not exist. Our
algorithms (like most other algorithms in the literature) work by maintaining
a linear sketch $L:\mathbb R^n\to\mathbb R^m$.
When computing the space requirement of such a
streaming algorithm, we assume all the updates are integers ($u\in\mathbb Z$)
and the coordinates of the vector $x$ throughout the stream remain bounded by
some value $M=\poly(n)$. We
make sure that the matrix of $L$ has also polynomially bounded
integer entries, this way maintaining $L(x)$ requires updating $m$ integer
counters and requires $O(m\log n)$ bits with fast
update time (especially since the matrices we consider are sparse). This
discretization step is standard and thus we omit most details.

In the standard model for randomized streaming algorithms the random bits used
(to generate the random linear map $L$, for example) are part of the space
bound. In contrast, our lower bounds do not make any assumption on the working of the
streaming algorithm and allow for the {\em random oracle model}, where the
algorithm is allowed free access to a random string at any time. All lower
bounds are proved through reductions from communication problems.

We say an event happens with {\em low probability} if the probability can
be made less than $n^{-c}$. Here $c>0$ is an arbitrary constant, for example
one can set $c=2$. The actual value of $c$ has limited effect on the space
of our algorithm: it changes only the unspecified constants hidden in the $O$
notation. We will routinely ignore low probability events, sometime even
$O(n)$ of them, which is okay as we leave $c$ unspecified.
Events complementary to low probability events are referred to as {\em high
probability} events.

For $0\le m\le n$ we call the vector $x\in\mathbb R^n$ {\em $m$-sparse} if
all but at most $m$ coordinates of $x$ are zero. We define
$\err_2^m(x)=\min\|x-\hat x\|_2$, where $\hat x\in\mathbb R^n$ ranges over all
the $m$-sparse vectors.



%
%
\section{The $L_p$ Sampler}\label{sec:lpsamp}
In this section, we present our $L_p$ sampler algorithm. In the following,
we assume  $p \in (0,2)$. This particular method does not seem to be
applicable for the $p=2$ case and we know of no $O(\log^2 n)$ space
$L_2$-sampling algorithm. We treat the $p=0$ case separately later.

We start by stating the properties of the two streaming algorithms we are
going to use. Both are based on maintaining $L(x)$
for a well chosen random linear map
$L:\mathbb R^n\to\mathbb R^{n'}$ with $n'<n$.

The {\em count-sketch} algorithm \cite{CharikarCF04} is so simple we cannot resist the
temptation to define it here. For parameter $m$, the count-sketch algorithm
works as follows. It selects independent samples $h_j:[n]\to[6m]$ and
$g_j:[n]\to\{1,-1\}$ from pairwise independent uniform hash families for
$j\in[l]$ and $l=O(\log n)$. It computes the following linear function of $x$
for $j\in[l]$ and $k\in[6m]$:
$y_{k,j}=\sum_{i\in[n],h_j(i)=k}g_j(i)x_i$. Finally it outputs
$x^*\in\mathbb R^n$ as an approximation of $x$ with
$x^*_i=\hbox{median}(g_j(i)y_{h(i),j}:j\in[l])$ for $i\in[n]$.

The performance guarantee of the count-sketch algorithm is as follows.
(For a compact proof see a recent survey by Gilbert and Indyk \cite{GilbertI10}.) 

\begin{lemma}\label{c-s}{\rm\cite{CharikarCF04}}
For any $x\in\mathbb R^n$ and $m\ge1$ we have
$|x_i-x^*_i|\le\err_2^m(x)/m^{1/2}$
for all $i\in[n]$ with high probability,
where $x^*$ is the output of the count-sketch algorithm with parameter $m$.
As a consequence we also have
$$\err_2^m(x)\le\|x-\hat x\|_2\le10\err_2^m(x)$$
with high probability, where $\hat x$ is the
$m$-sparse vector best approximating $x^*$ (i.e., $\hat x_i=x^*_i$ for the
$m$ coordinates $i$ with $|x^*_i|$ highest and is $\hat x_i=0$ for the
remaining $n-m$ coordinates).
\end{lemma}

We will also need the following result for the estimation of $L_p$ norms.

\begin{lemma}\label{norm}{\rm\cite{KaneNW10}}
For any $p\in(0,2]$ there is a streaming algorithm based on a random linear
map $L:\mathbb R^n\to\mathbb R^l$ with $l=O(\log n)$ that outputs a value $r$
computed solely from $L(x)$ that satisfies
$\|x\|_p\le r\le2\|x\|_p$
with high probability.
\end{lemma}

Our streaming algorithm on Figure~1 makes use of a single count-sketch and two
norm estimate algorithms. The count-sketch is for the randomly scaled version
$z$ of the vector $x$. One of the norm approximation algorithms is for
$\|x\|_p$, the other one approximates $\err_2^m(z)$ through the almost equal
value $\|z-\hat z\|_2$. A standard $L_2$ approximation for $z$ works if we
modify $z$ by subtracting $\hat z$ in the recovery stage. One can get
arbitrary good approximations of $\err_2^m(x)$ this way.

\begin{figure}
\fontsize{9}{11}
  \selectfont
\fbox{
\begin{minipage}[b]{.95\textwidth}

{\bf Initialization Stage:} \\
1. For $0<p<2$, $p\ne1$ set $k=10\lceil1/|p-1|\rceil$ and
$m=O(\epsilon^{-\max(0,p-1)})$ with a large\\
~~~~~~~~~~~~~ enough constant factor.\\
2. For $p=1$ set $k=m=O(\log(1/\epsilon))$ with a large enough constant factor.\\
3. Set $\beta=\epsilon^{1-1/p}$ and $l=O(\log n)$ with a large enough constant factor.\\
4. Select $k$-wise independent uniform scaling factors
$t_i\in[0,1]$ for $i\in[n]$.\\
5. Select the appropriate random linear functions for the execution of the
count-sketch\\ algorithm and $L$ and $L'$ for the norm estimations in the processing stage.\\

{\bf Processing Stage:}\\
1. Use count-sketch with parameter $m$ for the scaled vector $z\in\mathbb R^n$
with $z_i=x_i/t_i^{1/p}$.\\
2. Maintain a linear sketch $L(x)$ as needed for the $L_p$ norm approximation
of $x$.\\
3. Maintain a linear sketch $L'(z)$ as needed for the $L_2$ norm estimation of
$z$.\\

{\bf Recovery Stage:} \\
1. Compute the output $z^*$ of the count-sketch and its best $m$-sparse
approximation $\hat z$.\\
2. Based on $L(x)$ compute a real $r$ with $\|x\|_p\le r\le2\|x\|_p$.\\
3. Based on $L'(z-\hat z)=L'(z)-L'(\hat z)$ compute a real $s$ with $\|z-\hat
z\|_2\le s\le2\|z-\hat z\|_2$.\\
4. Find $i$ with $|z^*_i|$ maximal.\\
5. If $s>\beta m^{1/2}r$ or $|z^*_i|<\epsilon^{-1/p}r$ output FAIL.\\
6. Output $i$ as the sample and $z^*_it_i^{1/p}$ as an approximation for
$x_i$.
\end{minipage}
}
\label{fig:lpsampler}
\caption{
Our approximate $L_p$-sampler with both success probability and relative error $\Theta(\epsilon)$}
\vspace{-5mm}
\end{figure}

First we estimate the probability that the algorithm aborts because $s>\beta
m^{1/2}r$. This depends on the scaling that resulted in $z$ and it will be important
for us that the bound holds even after conditioning on any one scaling factor.

\begin{lemma}\label{abort}
Conditioned on an arbitrary fixed value $t$ of $t_i$ for a single index
$i\in[n]$ we have $\Pr[s>\beta m^{1/2}r\mid t_i=t]=O(\epsilon+n^{-c})$.
\end{lemma}

\begin{proof}
First note that by Lemma~\ref{norm} we have $r\ge\|x\|_p$ and
$s\le2\|z-\hat z\|_2$ with high probability. By
Lemma~\ref{c-s} we have $\|z-\hat z\|\le10\err_2^m(z)$ also
with high probability. We may therefore assume that all of these inequalities
hold, and in particular
$r\ge\|x\|_p$ and $s\le20\err_2^m(z)$. It is therefore enough to bound the
probability that $20\err_2^m(z)>\beta m^{1/2}\|x\|_p$.

For simplicity (and without loss of generality) we assume that the fixed
scalar is $t_n=t$ and will freely use $i$ for indexes in $[n-1]$.

Let $T=\beta\|x\|_p$. For each $i\in[n-1]$ we define
two variables $z'_i$ and $z''_i$ determined by $z_i$ as follows. The indicator
variable $z'_i=1$ if $|z_i|>T$ and $0$ otherwise. We set
$z''_i=z_i^2(1-z'_i)/T^2\in[0,1]$. Let $S'=\sum_{i\in[n-1]}z'_i$ and
$S''=\sum_{i\in[n-1]}z''_i$. Note that $T^2S''=\|z-w\|_2^2$, where $w$ is
defined by $w_i=z_iz'_i$ for $i\in[n-1]$ and $w_n=z_n$. Here $w$ is
$(S'+1)$-sparse, so we have $\err_2^m(z)\le TS''^{1/2}$ unless $S'\ge m$.
It is therefore enough to bound the probabilities of the events
$S'\ge m$ and $S''>m\beta^2\|x\|_p^2/(20T)^2=m/400$, each with  $O(\epsilon)$.

We have $\E[z'_i]=|x_i|^p/T^p$, $\E[S']\le\beta^{-p}=\epsilon^{1-p}$. By our
choice of $m$ and the concentration of $S'$ provided by $k$-wise independence
we have $\Pr[S'\ge m]=O(\epsilon)$ as needed.
The calculation for $S''$ is similar. We have
$$\E[z''_i]<\int_{|x_i|^p/T^p}^\infty x_i^2t^{-2/p}T^{-2}dt=\frac
p{2-p}|x_i|^pT^{-p}.$$ Thus $\E[S'']\le\frac
p{2-p}\|x\|_p^pT^{-p}=O(\beta^{-p})=O(\epsilon^{1-p})$. Note that the $z''_i$
are not indicator variables as the $z'_i$, but they are still $k$-wise
independent random variables from $[0,1]$ and we can bound the probability of
large deviation for $S''$ as we did for $S'$. This completes the proof of the
lemma.
\end{proof}

The fact that our algorithm is an approximate $L_p$-sampler with both relative
error and success probability $\Theta(\epsilon)$ follows from the following
lemma. Indeed, if the probabilities were exactly $\epsilon|x_i|^p/r^p$ and
if $\|x\|_p\le r\le2\|x\|_p$ would always hold, we
would make no relative error and the success probability would be
$\E[\epsilon\|x\|_p^p/r^p]\ge\epsilon/2^p$.

\begin{lemma} \label{lps}
The probability that the algorithm of Figure~1 outputs the index
$i\in[n]$ conditioned on a fixed value for $r\ge\|x\|_p^p$ is
$(\epsilon+O(\epsilon^2))|x_i|^p/r^p+O(n^{-c})$. The
relative error of the estimate for $x_i$ is at most $\epsilon$ with high
probability.
\end{lemma}

\begin{proof} Optimally, we would output $i\in[n]$ if
$|z_i|>\epsilon^{-1/p}r$. This happens if $t_i<\epsilon|x_i|^p/r^p$ and
has probability exactly $\epsilon|x_i|^p/r^p$. We have to estimate the
probability that something goes wrong and the algorithm outputs $i$ when this
simple condition is not met or vice versa.

Three things can go wrong. First, if $s>m^{1/2}\beta r$ the algorithm
fails. This is only a problem for our calculation if it should, in fact,
output the index $i$. Lemma~\ref{abort} bounds the conditional
probability of this happening.

Having dealt with the $s>\beta m^{1/2}r$ case we may assume now $s\le\beta
m^{1/2} r$. We also make the assumptions (high probability by
Lemma~\ref{norm}) that
$\|z-\hat z\|_2\le s$ and thus $\err_2^m(z)\le\|z-\hat z\|_2\le s\le\beta
m^{1/2}r$. Finally, we also assume $|z^*_i-z_i|\le\err_2^m(z)/m^{1/2}\le\beta
r$ for all $i\in[n]$. This is satisfied with high probability by
Lemma~\ref{c-s}.

A second source of error comes from this $\beta r$ possible difference
between $z^*_i$ and $z_i$. This can only make a problem if $t_i$ is
close to the threshold, namely $(\epsilon^{-1/p}+\beta)^{-p}|x_i|^p/r^p\le
t_i\le (\epsilon^{-1/p}-\beta)^{-p}|x_i|^p/r^p$. The probability of selecting
$t_i$ from this interval is
$O(\beta/\epsilon^{1+1/p}|x_i|^p/r^p)=O(\epsilon^2|x_i|^p/r^p)$ as required.

Finally, the third source of error comes from the possibility that $i$ should
be output based on $|z_i|>\epsilon^{-1/p}r$, yet we output another index
$i'\ne i$ because $z^*_{i'}\ge z^*_i$. In this case we
must have $t_{i'}<(\epsilon^{-1/p}-\beta)^{-p}|x_i|^p/r^p$. This has
probability $O(\epsilon|x_{i'}|^p/r^p)$. By the union bound the probability
that such an index $i'$ exists is
$O(\epsilon\|x\|_p^p/r^p)=O(\epsilon)$. Pairwise independence is enough to
conclude that the same bound holds after conditioning on
$|z_i|>\epsilon^{-1/p}r$. This finishes the proof of the first statement of
the lemma.

The algorithm only outputs an index $i$ if $s\le\beta m^{1/2}r$ and
$|z^*_i|\le\epsilon^{-1/p}r$. The first implies that the absolute
approximation error for $z_i$ is at most $\beta r$, while the second lower
bounds the absolute value of the approximation itself by $\epsilon^{-1/p}r$,
thus ensuring a $\beta\epsilon^{1/p}=\epsilon$ relative error
approximation. Our approximation for $x_i=z_it_i^{1/p}$ is $z^*_it^{1/p}$, so
the relative error is the same. Note that the
inverse polynomial error probability comes from the various low probability
events we neglected. The same is true for the additive error term in the
distribution.
\end{proof}

\begin{theorem}\label{thm:sampler} For $\delta>0$ and $\epsilon>0$, $0<p<2$ 
 there is an $O(\epsilon)$ relative
  error one pass $L_p$-sampling algorithm with failing probability at most $\delta$ and having
  low probability that the relative error of the estimate for the selected
  coordinate is more than $\epsilon$. 
    The algorithm uses
  $O_p(\epsilon^{-\max(1,p)}\log^2n\log(1/\delta))$ space for $p\ne1$ while
  for $p=1$ the space is
  $O(\epsilon^{-1}\log(1/\epsilon)\log^2n\log(1/\delta))$.
\end{theorem}

\begin{proof}
Using Lemma~\ref{lps} and the fact that $\|x\|_p\le r\le2\|x\|_p$ with high
probability one obtains that the failure probability of the algorithm in Figure~1 
is at most $1-\epsilon/2^p+O(n^{-c})$. Conditioning on
obtaining an output, returning $i$ has probability
$(1+O(\epsilon))|x_i|^p/\|x\|_p^p+O(n^{-c})$. Clearly, the latter statement
remains true for any number of repetitions and the failure probability is
raised to power $v$ for $v$ repetitions. Thus using
$v=O(\log(1/\delta)/\epsilon)$ repetitions (taking the first non-failing output),
 the algorithm is an $O(\epsilon)$
relative error $\delta$ failure probability $L_p$-sampling algorithm. Here we  
assume $v<n$ as otherwise recording the entire vector $x$ is more efficient.

The low probability of more than $\epsilon$ relative error in estimating $x_i$
also follows from Lemma~\ref{lps}.
In one round, the algorithm on Figure~1 uses
$O(m\log n)$ counters for the count-sketch and this dominates the
counters for the norm estimators. Using standard discretization this can be
turned into an $O(m\log^2n)$ bit algorithm. For the discretization we also
have to keep our scaling factors polynomial in $n$. Recall that in the
continuous model these factors $t_i^{-1/p}$ were unbounded. But we can safely
declare failure if $t_i^{-1}>n^c$ for some $i\in[n]$ as this has low
probability $n^{1-c}$. We have to do the $v$ repetitions of the algorithm
in parallel to obtain a single pass streaming algorithm. This increases the
space to $O(vm\log^2n)$ which is the same as the one claimed in the theorem.
\end{proof}

Note that the hidden constant in the space bound of the theorem depends on
$p$, especially that $1/(2-p)$, $1/p$ and $1/|1-p|$ factors come in. The last
can
always be replaced by a $\log(1/\epsilon)$ factor but the former ones are
harder to handle. For $p=2$ an extra $\log n$ factor seems to be necessary for
an algorithm along these lines, see \cite{AndoniKO10}.

As we will see in Theorem~\ref{lpl}, our space bound is tight for $\epsilon$
and $\delta$ constants. Note that the lower bound holds also if we only
require the overall distribution of the $L_p$-sampler to be close to the $L_p$
distribution as opposed to the much more strict definition of $\epsilon$
relative error sampling.

%
%
\subsection{The $L_0$ Sampler}
  For $p$ near zero, the method of precision sampling
   becomes intractable. This is because our scaling factors are
   $t_i^{-1/p}$ which clearly rules out $p=0$.
   In the following we present a $L_0$ using a different approach.
First we state the following well-known result on exact recovery of sparse vectors.
\begin{lemma}\label{lem:sparse}
For $1\le s\le n$ and $k=O(s)$ there is a random linear function
$L:\mathbb R^n\to\mathbb R^k$ (generated from $O(k\log n)$ random bits) and a
recovery procedure that on input $L(x)$ outputs $x'\in\mathbb R^n$ or DENSE,
satisfying that for any $s$-sparse $x$ the output is $x'=x$ with
probability $1$, otherwise the output is DENSE with
high probability.
\end{lemma}

\begin{theorem}\label{l0}
There exists a zero relative error $L_0$ sampler which
 uses $O(\log^2 n\log(1/\delta))$ bits and outputs a 
 coordinate $i\in[n]$ with probability at least $1-\delta$.
\end{theorem}
\begin{proof} We first present our algorithm assuming a random oracle, and
  then we remove this assumption through the use of the pseudo-random
  generator of Nisan \cite{Nisan}. Let $I_k$ for $k=1,\ldots,\lfloor\log
  n\rfloor$ be subsets of $[n]$ of size $2^k$ chosen uniformly at random and
  $I_0=[n]$. For each
  $k$ we run the sparse recovery procedure of
  Lemma \ref{lem:sparse} on the vector $x$ restricted to the coordinates in
  $I_k$ with $s$ set to $\lceil4\log(1/\delta)\rceil$. We return a uniform
  random non-zero coordinate from the first recovery that gives a non-zero
  $s$-sparse vector. The algorithm fails if each recovery algorithm returns
  zero or DENSE.

  Let $J$ be the set of coordinates $i$ with $x_i\ne0$ (the support of $x$).
  Disregarding the low probability error of the procedure in
  Lemma~\ref{lem:sparse} this procedure returns each index $i\in J$ with equal
  probability and never returns an index outside $J$. To bound the failure
  probability we observe that for $|J|\le s$ failure is not possible, while for
  $|J|>s$ one has $k\in[\lfloor\log n\rfloor]$ such that $\E[|I_k\cap
  J|]=2^k|J|/n$ is between $s/3$ and $2s/3$. For this $k$ alone $1\le|I_k\cap
  J|\le s$ is satisfied with probability at least $1-\delta$ by the Chernoff
  bound limiting failure probability by $\delta$.

  To get rid of the random oracle we use Nisan's generator \cite{Nisan} that
  produces the random bits for the algorithm (including the ones describing
  $I_k$ and the ones for the eventual random choice from $I_k\cap J$) from an
  $O(\log^2 n)$ length seed. It fools every logspace tester including the one
  that tests for a fixed set $J\subseteq[n]$ and $i\in[n]$ if the algorithm
  (assuming correct reconstruction) would return $i$. Thus this version of the
  algorithm, now using $O(\log^2n)$ random bits and $O(\log^2\log(1/\delta))$
  total space, is also a zero relative error $L_0$-sampler with failure
  probability bounded by $\delta+O(n^{-c})$.
\end{proof}

As we will see in Theorem~\ref{lpl}, this space bound is also tight for
$\delta$ a constant and better sampling is not possible even if we allow
constant relative error or a small overall distance of the output from the
$L_0$ distribution.

%
%
\section{Finding Duplicates}\label{sec:duplicates}
Recall that, given a data stream of length $n+1$ over the alphabet $[n]$, finding
duplicates problem asks to output some $a\in[n]$ that has appeared
at least twice in the stream


\begin{theorem}\label{thm:dupub}
For any $\delta>0$ there is a $O(\log^2 n\log(1/\delta))$ space one-pass
algorithm which, given a stream of length $n+1$ over the alphabet $[n]$,
outputs an $i\in[n]$ or FAIL, such that the probability of outputting FAIL
is at most $\delta$ and the algorithm outputs a letter $i\in[n]$ that is no
duplicate with low probability.
\end{theorem}
\begin{proof}
Let $x$ be an $n$-dimensional vector, initially zero at each coordinate. We
run the $L_1$-sampler of Theorem~\ref{thm:sampler} on $x$, with both relative error
and failure probability set to $1/2$. Before we start processing the
stream, we subtract 1 from each coordinate of $x$; i.e., we feed the updates
$(i,-1)$ for $i=1,\ldots n$ to the $L_1$ sampling algorithm. When a stream
item $i\in [n]$ comes, we increase $x_i$ by 1; i.e., we generate the update
$(i,1)$.

Observe that when the stream is exhausted, we have $x_i\geq 1$ for items $i$
that have at least two occurrences in the stream, $x_i=0$ for items that
appear once, and $x_i=-1$ for items that do not appear.  Note that our
$L_1$-sampler, if it does not fail, outputs an index $i$ and an approximation
$x^*$ of $x_i$. If $x^*$ is positive, we output $i$, if it is
negative or the $L_1$-sampler fails, we output FAIL. We have
$\sum_{i=1}^nx_i=1$,  hence a perfect $L_1$ sample from $x$ is positive with
more than half probability. Taking into account that our $L_1$-sampler has
$1/2$ relative error and failure probability (and neglecting for a second the
chance that $x^*$ has different sign from $x_i$) we conclude that we output a
duplicate with probability at least $1/4$. The event that $x^*$ does not have
the same sign as $x_i$ (and thus the relative error is at least 1) has low
probability. This low probability can increase the failure probability and/or
introduce error when we output non-duplicate items.

Repeating the algorithm $O(\log(1/\delta))$ times in parallel and accepting
the first non-failing output reduces the failure rate to
$\delta$ but keeps the error rate low.
\end{proof}

As we will see in Theorem~\ref{thm:duplb}, our space bound is tight for
$\delta<1$ a constant. 

It is natural to study the duplicates problem for other ranges of
parameters. Assume that we have a stream of length $n-s\le n$ over the
alphabet $[n]$. For this problem, Gopalan et al.\ \cite{GopalanJaikumar} gave
an $O((s+1)\log^3 n)$ bits algorithm and an $\Omega(s)$ lower bound. Here we
give an algorithm which uses $O(s\log n+\log^2 n)$ space.

\begin{theorem}\label{dups}
For any $\delta>0$ there is an $O(s\log n+\log^2 n\log
  1/\delta)$ space one-pass algorithm which, given a stream of length $n-s$
  over the alphabet $[n]$, outputs NO-DUPLICATE with probability 
  1 if the input sequence has no duplicates, otherwise 
  it outputs $i \in [n]$ or reports FAIL. The returned number is a
  duplicate with high probability while the probability of returning FAIL is at most $\delta$.
\end{theorem}
\begin{proof} Let $x$ be an $n$-dimensional vector updated
  according to the description in the proof of Theorem \ref{thm:dupub}; i.e.,
  $x_i$ is one less than the number of times $i$ appears in the stream. In parallel,
   we run the exact recovery procedure from Lemma \ref{lem:sparse} with parameter $5s$
  and the $1/2$ relative error $L_1$-sampler of 
  Theorem~\ref{thm:sampler} with failure rate $1/2$, both on the vector $x$. If the
  recovery algorithm returns a vector (as opposed to DENSE) we proceed 
  and give the correct output
  assuming we have learned the entire $x$. Otherwise we consider the output of the sampling
  algorithm. If it is $(i,x^*)$ with $x^*>0$ we report $i$ as a duplicate
  otherwise (if $x^*\le0$ or the sampling algorithm fails) we output FAIL.
Define
\begin{align*}
\|x\|^+_1=\sum_{i:x_i>0} |x_i| & &\text{ and }& &\|x\|_1^-=\sum_{i:x_i<0} |x_i|.
\end{align*}
Note that $\|x\|^+_1-\|x\|_1^-=\sum_{i=1}^n x_i = -s$.
 If $\|x\|^+_1 + \|x\|^-_1 \leq 5s$, then $x$ is $5s$-sparse, thus the sparse
recovery procedure outputs $x$ and the algorithm makes no error. Note that the
no repetition case falls into this category. If, however,
$\|x\|^+_1 + \|x\|^-_1 > 5s$, then the probability that a perfect $L_1$ sample
from $x$ is positive is $\|x\|^+_1/\|x\|_1 > 2/5$. Taking into account the
relative error and failing probability (but ignoring the low probability event
of the sampler outputting a wrong sign or sparse recovery algorithm reporting
a vector), we conclude that with probability at
least 1/10 we get a positive sample and a correct output, otherwise we output
FAIL. The failure probability can be decreased to $\delta$ by
$O(\log(1/\delta))$ independent repetitions of the sampler. Note that the
sparse recovery does not have to be repeated as it has low error probability.

The sparse recovery procedure takes $O(s\log n)$ bits by
Lemma~\ref{lem:sparse} for $s>0$ (it takes $O(\log n)$ bits for $s=0$) and
each instance of the $L_1$-sampler requires $O(\log^2 n)$ bits by
Theorem~\ref{lps}, totaling $O(s\log n + \log^2n\log1/\delta)$ bits. 
\end{proof}

Here we do not have a matching lower bound, but only the $\Omega(\log^2n+s)$ that
follows from the $\Omega(s)$ bound in \cite{GopalanJaikumar} and our
$\Omega(\log^2 n)$ bound on the original version of the duplicates problem.

We remark the last two theorems can be stated in a bit more general
form. Instead of considering repetitions in data streams one can consider the
problem of finding an index $i$ with $x_i>0$ for a vector
$x\in\mathbb Z^n$ given by an update stream. Let
$s=-\sum_{i=1}^nx_i$. If $s<0$, then a positive coordinate exists and the
algorithm of Theorem~\ref{thm:dupub} finds one using $O(\log^2 n\log(1/\delta))$
space with low error and at most $\delta$ failure probability.
If $s\ge0$ a positive coordinate does not necessarily exist, but the algorithm
of Theorem~\ref{dups} finds one, report none exists or fails with the error,
failure and space bounds claimed there.

Let us consider finally the version of the duplicates problem, where we have a
stream of length $n+s>n$ over the alphabet $[n]$. Our lower and upper bounds
are even farther in this case. A duplicate can be found using $O(\min\{\log^2n,
(n/s)\log n\})$ bits of memory in one pass with constant probability as
follows. If we sample a random item from the stream, it will appear again
unless that was the last appearance of the letter. As there are at most $n$
last appearances in the stream of length $n+s$, the probability for a uniform
random sample to repeat later is at least $s/(n+s)$. Therefore, if $n/s < \log
n$, we can sample $4\lceil n/s\rceil$ items from the stream uniformly at
random and check if one of them appears again to obtain a constant error
algorithm for finding duplicates. If on the other hand $n/s \geq \log n$, we use
the algorithm in Theorem~\ref{thm:dupub}.

Combining our lower bound for the original version of the duplicates problem 
with the simple lower bound of  $\Omega(\log n)$, we conclude that any 
streaming algorithm that finds a duplicate in length $n+s$ streams must use 
$\Omega(\log^2(n/s)+ \log n)$ bits.

%
%
\section{Lower Bounds}\label{sec:lb}
All our lower bounds follow from the augmented indexing problem. This problem
is defined as follows. Let $k$ and $m$ be positive integers.
The first player Alice, is given a string $x\in[k]^m$, while the second player Bob
is given an integer $i\in [m]$ and $x_j$ for $j<i$. Alice sends a single
message to Bob and Bob should output $x_i$. 

\begin{lemma}[\cite{MiltersenNSW}]\label{lem:ai}
In any one-way protocol in the joint random source model with success
probability at least $1-\delta>\frac{3}{2k}$, Alice must send a message of
size $\Omega((1-\delta)m\log k)$.
\end{lemma}

We will use this lemma by reducing augmented indexing to other communication
or streaming problems.

%
%
\subsection{Universal Relation}\label{sec:ur}
Consider the following two player communication game.
Alice gets a string $x\in\{0,1\}^n$, Bob gets $y\in\{0,1\}^n$ with
  the promise that $x\neq y$. The players exchange 
  messages and the last player to receive a message must
   output an index $i\in [n]$ such that $x_i\neq y_i$. 
  We call this the {\em universal relation communication problem} and denote it by $\URn$.

This relation has been studied in detail for deterministic communication, as 
it naturally arises in the context of
 Karchmer-Wigderson games \cite{KarchmerWigderson}. We note however
 that our definition is slightly unusual: in most settings both players must
 obtain the same index $i$ such that $x_i\neq y_i$, whereas we are satisfied
 with the last player to receive a message learning such an $i$. Clearly, the
 stronger requirement can be met in $\lceil\log n\rceil$ additional bits and
 one additional round. The additional bits are needed in deterministic case
 but we are not concerned with $O(\log n)$ terms for our bounds, so the two
 models are almost equivalent up to the shift of one in the number of rounds.

The best deterministic protocol for $\URn$ is due to 
 Tardos and Zwick~\cite{TardosZwick}. Improving a previous result 
 by Karchmer \cite{Karchmer}, they gave a 3 round deterministic protocol 
 using $n+2$ bits of communication with both players learning the same index
 $i$ and showed that $n+1$ bits is necessary for such a protocol. They also
 gave an $n-\lfloor\log n\rfloor+2$ bit 2 round deterministic protocol for our
 weaker version of the problem, which is also tight except for the $+2$
 term. They also gave an $n-\lfloor\log n\rfloor+4$ bit 4 round protocol, where
 both players find an index where $x$ and $y$ differ---but not necessarily
 the same index. This shows that finding the same difference is harder.

Let $R^k_\delta(U)$ denote the $k$-round $\delta$-error communication
complexity of the communication problem $U$. We write $R_\delta(U)$ to denote
the $\delta$-error communication complexity when the number of rounds is not
bounded. The following proposition follows from similar ideas that were used in 
Theorem~\ref{l0}. See the Appendix for a sketch of the proof.

\begin{proposition}\label{thm:urub}
It holds that $R^1_\delta(\URn)=O(\log^2 n\log\frac{1}{\delta})$ and $R^2_\delta(\URn)=O(\log n\log\frac{1}{\delta})$.
\end{proposition}

We remark that along similar lines one can find an $O(\log n \log\log n\log1/\delta )$
space two-pass zero relative error $L_0$-sampling algorithm, by estimating  $L_0$ 
of the vector defined by the stream in the first pass using \cite{KaneNW10}. Next
 we will show that the above proposition is best possible up to the $O(\log\delta^{-1})$
terms. We start with an averaging lemma. The proof can be found in the Appendix.

\begin{lemma}\label{aver} Any protocol for $\URn$ can be turned into one that
outputs every index $i\in[n]$ with $x_i\ne y_i$ with the same probability. The
new protocol uses a joint random source. The number of bits sent, the number
of rounds and the error probability does not change.
\end{lemma}

\begin{theorem}\label{thm:urlb}
For any $\delta<1$ constant we have $R^1_\delta(\URn)=\Omega(\log^2 n)$ and
$R_\delta=\Omega(\log n)$.
\end{theorem}

\begin{proof}
The second bound comes from considering a uniform random pair $(x,y)$ with
Hamming distance 1. Either player needs to get $\log n$ bits of information to
learn the only index where the strings differ.

To prove the first bound suppose Alice and Bob wants to
solve the augmented indexing problem with Alice receiving $z\in[2^t]^s$ and Bob
getting $i\in [s]$ and $z_j$ for $j<i$.

Let them construct real vectors $u$ and
$v$ as follows. Let $e_q\in\mathbb R^{2^t}$
be the standard unit vector in the direction of coordinate $q$. Alice forms
the vectors $w_j$ by concatenating $2^{s-j}$ copies of $e_{z_j}$, then she
forms $u$ by concatenating these vectors $w_j$ for
$j\in[s]$. The dimension of $u$ is $n=(2^s-1)2^t$. Bob
obtains $v$ by concatenating the same vectors $w_j$ for $j\in[i-1]$ (these are
known to him) and then
concatenating enough zeros to reach the same dimension
$n$.

Now Alice and Bob perform the $R^1_\delta(\URn)$ length
$\delta$ error one round protocol for $\URn$. By Lemma~\ref{aver} we
may assume the protocol returns a uniform random index where $u$ and $v$
differ. Note that each such index reveals one coordinate $z_j\in[2^t]$ to Bob
for $j\ge i$. As $z_j$ is revealed by $2^{s-j}$ such indices more than half the
time when the $\URn$ protocol does not err Bob learns the correct value of
$z_i$. This yields a $R^1_\delta(\URn)$ length one way protocol for augmented
indexing with error probability $(1+\delta)/2$. By Lemma~\ref{lem:ai} we have
$R^1_\delta(\URn)=\Omega(st)$. Choosing $s=t$ proves the theorem.
\end{proof}

\subsection{Finding Duplicates}\label{sec:dublb}
\begin{theorem}\label{thm:duplb}
Any one-pass streaming algorithm that outputs a duplicate with constant
probability  uses $\Omega(\log^2 n)$ space. This remains true even if the
stream is not allowed to have an element repeated more than twice. 
\end{theorem}
\begin{proof}
We show our claim by a reduction from the universal relation. Each of
 Alice and Bob is given a binary string of length $n$, respectively $x$ 
and $y$. Further, the players are guaranteed that $x\neq y$. Alice 
sends a message to Bob, after which Bob must output an index 
$i\in[n]$ such that $x_i\neq y_i$. By Theorem~\ref{thm:urlb}, to solve this 
problem with $1/2$ error probability requires $\Omega(\log^2 n)$ bits for one-way communication. 
Alice constructs the set $S=\{2i-1+x_i\mid i\in[n]\}\subseteq[2n]$ and Bob
constructs $T=\{2i-y_i\mid i\in[n]\}\subseteq[2n]$. Observe that $|S|=|T|=n$ 
and $x_i\neq y_i$ if and only if either $2i$ or $2i-1$ is in both $S$ and
$T$.

Next, using the shared randomness, players pick a random subset 
$P$ of $[2n]$ of size $n$. We have
$$\Pr[|S\cap P| + |T\cap P|\geq n+1]>1/8.$$
To see this let $i\in S\cap T$ and $j\in[2n]\setminus(S\cap T)$. We have
$|P\cap\{i,j\}|=1$ with probability more than $1/2$. The sets $P$ satisfying
this can be partitioned into classes of size four by putting $Q\cup\{i\}$,
$Q\cup\{j\}$ and their complements in the same class for any
$Q\subseteq[2n]\setminus\{i,j\}$, $|Q|=n-1$. Clearly, at least one of the four
sets $P$ in each class satisfies $|S\cap P|+|T\cap P|>n$.

Given a streaming 
algorithm $A$ for finding duplicates, Alice feeds the elements of 
$S\cap P$ to $A$ and sends the memory contents over to Bob, 
along with the integer $|S\cap P|$. If $|S\cap P|+|T\cap P|<n+1$, 
Bob outputs FAIL. Otherwise, feeds arbitrary $n+1-|S\cap P|$ 
elements of $T\cap P$ to $A$. Note that no element repeats more than twice.

On the other hand $|P|=n$ and we always give $n+1$ elements of $P$ 
to the algorithm. Also with constant probability, Bob finds an 
$a\in S \cap T$, which in turn reveals an $i$ such that $x_i\neq y_i$. 
Therefore by Theorem \ref{thm:urlb}, any algorithm for finding 
duplicates must use $\Omega(\log^2 n)$ bits.
\end{proof}

\subsection{$L_p$-sampling}
Our algorithm for the duplicates problem (Theorem~\ref{thm:dupub}) is based on
$L_1$-sampling, thus the matching lower bound for the duplicates problem
implies a similar matching bound for the sampling problem. We state this
result here. Notice that the $L_p$ distribution corresponding to $0,\pm1$
vectors are independent of $p$, so $p$ does not have to be specified for the
next theorem.

\begin{theorem}\label{lpl}
Any one pass
$L_p$-sampler with an output distribution,
whose variation distance from the $L_p$ distribution corresponding to $x$ is
at most $1/3$, requires $\Omega(\log^2n)$ bits of memory. This holds even
when all the coodinates of $x$
are
guaranteed to be $-1$, $0$ or $1$.

For constants $\delta<1$ and $\epsilon<1$ the same lower bound holds for any
$\epsilon$ relative error $L_p$-sampler with failure probability $\delta$.
\end{theorem}

\begin{proof}
Consider the $L_1$ sampling algorithm that we used to prove
Theorem~\ref{thm:dupub}. Given a stream of items from $[n]$ we turned it to
an update stream for an $n$ dimensional vector $x$ by first
producing an update $(i,-1)$ for all $i\in[n]$ and then for any letter $i$ in
the stream producing an update $(i,1)$. Assuming that no item appears more
than twice in the stream all coordinates of the final vector $x$ are $-1$, $0$
or $1$. The $L_1$ distribution for $x$ puts weight more than $1/2$ on the
coordinates having value $1$. These are the duplicates. Thus if we have
another algorithm such that the variation distance of its output is at most
$1/3$ from this $L_1$ distribution, then it returns a coordinate with value 1
with probability at least $1/6$. For an $\epsilon$ relative error $\delta$
failure probability approximate $L_p$-sampler the same probability is at least
$(1-\epsilon)(1-\delta)-n^{1-c}$. Finding a coordinate in $x$ with value $1$
is the same as finding a duplicate in the original stream, so we need
$\Omega(\log^2n)$ memory by Theorem~\ref{thm:duplb}.
\end{proof}

\subsection{Heavy Hitters}\label{sec:hh}
The heavy hitters problem in the streaming model is defined as follows. Let
$x$ be an $n$-dimensional integer vector given by an update stream.
A heavy hitters algorithm with parameters $p>0$ and $\phi>0$ is
required to output a set $S\subseteq [n]$ that contains
all $i$ with $|x_i|\geq \phi\|x\|_p$ and no $i$ such
that $|x_i|\leq \frac{\phi}{2}\|x\|_p$. We call such $S$ a valid heavy hitter
set.\footnote{ In general, the parameter $\frac12\phi$
can be replaced by any $\epsilon < \phi$. 
 Since here our 
focus is on lower bounds, we have simplified the definition.} 
In this part, we show a tight lower bound for the
 space complexity of randomized algorithms (assuming
constant probability of error) for 
the heavy hitter problem. First we briefly review the upper bounds.

The count-median algorithm from \cite{CormodeM05} gives
 a $O(\phi^{-1}\log^2 n)$ space upper bound for the case of $p=1$.
Here we note the count-sketch \cite{CharikarCF04} in fact 
gives a $O(\phi^{-p}\log^2 n)$ space upper bound for all $p \in (0,2]$.
The case of $p=2$ easily follows from Lemma~\ref{c-s}. Let $d=\err^m_2(x)/m^{1/2}$.
In general it holds $d \le ||x||_p/m^{1/p}$ for any
$p\in(0,2]$. Indeed, let $H \subset [n]$ be the set of indices for which
$ d^2=\sum_{i\in H}x_i^2/m$ and let $c=\max_{i\in H}|x_i|$. Then we have
$||x||_p^p/m=\sum_{i\in[n]}|x_i|^p/m\ge c^p+\sum_{i\in H}|x_i|^p/m\ge
c^p+c^{p-2}\sum_{i\in
H}x_i^2/m=c^p+c^{p-2}d^2\ge c^p((1-p/2)+(p/2)c^{-2}d^2\ge
c^p(c^{-2}d^2)^{p/2}=d^p.$ Therefore setting $m=1/\phi^p$ in the count-sketch
scheme gives the desired result.

We remark that a similar upper bound for the heavy hitter problem is 
shown in \cite{KaneNPW} (cf. Theorem 1), albeit via different arguments.  
%
%
In the next theorem, we show that the above upper bound is tight for
any reasonable range of parameters. Our lower bound holds even in the strict
turnstile model and even for very short streams.

\begin{theorem}\label{hhl} Let $p>0$ and $\phi\in(0,1)$ be a reals. Any one
  pass heavy hitter algorithm in the strict turnstile model uses
  $\Omega(\phi^{-p}\log^2n)$.
\end{theorem}

\begin{proof} Suppose there is a one pass heavy hitter algorithm for
  parameters $p$ and $\phi$. We allow for a
  random oracle and assume the updates are polynomially bounded in $n$ and
  integers. We can also restrict the number of updates to be $O(\phi^{-p}\log
  n)$ and assume all coordinates of the final vector are positive (strict
  turnstile model). We turn this streaming algorithm into a protocol for
  augmented indexing in a similar way as we transformed the protocol for
  $\URn$ to a protocol for augmented indexing in the proof of
  Theorem~\ref{thm:urlb}. The exponential growth is now achieved not by
  repetition but by multiplying the coordinates with a growing factor.

Suppose Alice and Bob wants to
solve the augmented indexing problem and Alice receives $y\in[2^t]^s$ and Bob
gets $i\in [s]$ and $y_j$ for $j<i$. Let them construct real vectors $u$ and
$v$ as follows. Let $b=(1-(2\phi)^p)^{-1/p}$ and let $e_q\in\mathbb R^{2^t}$
be the standard unit vector in the direction of coordinate $q$. Alice obtains
$u$ by concatenating the vectors $\lceil b^{s-j}\rceil e_{y_j}$ for
$j\in[s]$. The dimension of $u$ is $n'=s2^t$. Bob
obtains $v$ by concatenating the same vectors for $j\in[i-1]$ and then
concatenating enough zeros, namely $(s-i+1)2^t$, to reach the same dimension
$n'$. Now Alice and Bob perform the heavy hitter algorithm for the vector
$x=u-v$ as follows. Alice generates the necessary updates to increase the
initially zero vector $x\in\mathbb Z^n$ to reach $x=u$, maintains the memory
content throughout these updates and sends the final content to Bob. Now Bob
generates the necessary updates to decrease $x=u$ to its final value $x=u-v$
and maintains the memory throughout. Finally Bob learns the heavy hitter set
$S$ the streaming algorithm produces and outputs $z\in[2^t]$ if the smallest
index in $S$ is $(i-1)2^t+z$.

We claim that the above protocol errs only if the streaming algorithm makes an
error. Notice that all coordinates of $x_l$ of $x=u-v$ are zero except the
ones of the form $x_{l_j}=\lceil b^{s-j}\rceil$ for $l_j=(j-1)2^t+y_j$, where
$i\le j\le s$. Thus $x_{l_i}$ is the first non-zero coordinate. So the claim
is true if $x_{l_i}\ge\phi\|x\|_p$. Using $\lceil v\rceil<2v$ for $v\ge1$ we
get exactly this:
\begin{align*}
\phi^p\|x\|_p^p &= \phi^p\sum_{j=i}^s\lceil b^{s-j}\rceil^p\\
&<(2\phi)^pb^{p(s-i+1)}/(b^p-1)\\
&=b^{p(s-i)} \quad\qquad\text{(since $b^p =1/(1-(2\phi)^p)$)}\\
&\le x_{l_i}^p
\end{align*}

Let us now choose $s=\lceil(2\phi)^{-p}\log n\rceil$ and $t=\lceil\log
n/2\rceil$. For large enough $n$ this gives $n'=s2^t<n$ and all coordinates of
$x$ throughout the procedure remain under $n$. Still if the streaming
algorithm works with probability over 1/2, then  by Lemma~\ref{lem:ai} the message
size of the devised protocol is $\Omega(st)=\Omega(\phi^{-p}\log^2n)$. This
proves the theorem as the message size of the protocol is the same as the
memory size of the streaming algorithm.
\end{proof}


{\small
\bibliography{p17}{}}
\bibliographystyle{plain}

%
%
\newpage
\appendix

\section{Appendix}\label{sec:appx}

\subsection{Missing proofs}

\begin{proof}[Proof of Proposition~\ref{thm:urub}] (sketch) One way to deduce the one round protocol is from
  Theorem~\ref{l0}. Alice and Bob run a single pass $L_0$-sampling algorithm
  on $x-y$. This can be achieved by a single message from Alice to Bob
  containing the memory after the first set of updates as in the proof of
  Theorem~\ref{hhl}. The sample $i$ Bob finds is an (almost uniform random)
  index with $x_i\ne y_i$.

Looking more closely to this algorithm we have presented, it finds an index
where $x$ and $y$ disagree from some set $I\subseteq[n]$ that contains at
least one, but not too many such indices. It tries $O(\log n)$ random sets so
that one of them works. One can obtain the two round
protocol by finding such a set in the first round and concentrating on a
single such set in the second round.
\end{proof}

\begin{proof}[Proof of Lemma~\ref{aver}] Using the joint random source the players take a uniform random
  permutation $\pi$ of $[n]$ and use it to permute the digits of $x$ and
  $y$. Further they take a uniform random subset $S\subseteq[n]$ and flip the
  digits with coordinates in $S$. This requires no communication. 
Then they run the original protocol on the modified inputs and report
$\pi^{-1}(i)$ if the original protocol reports $i$. 
All indices where $x$ and $y$ differ are reported with equal probability by
symmetry.
\end{proof}

\end{document}